\documentclass[lettersize,11pt,abstract]{scrartcl}

\usepackage[margin=1in]{geometry}
\usepackage{settings}
\usepackage[]{todonotes}
\usepackage{algorithm}
\usepackage[noend]{algpseudocode}
\usepackage{setspace}
\usepackage{tikz}
\usetikzlibrary{positioning, decorations}

\usepackage{mathdots}

\newcommand{\bZ}{\mathbb{Z}}

\newcommand{\cO}{\mathcal{O}}

\usepackage{bm}
\newcommand{\vect}[1]{\bm{#1}}

\renewcommand\epsilon\varepsilon

\DeclarePairedDelimiter\abs{\lvert}{\rvert}

\DeclarePairedDelimiter\floor{\lfloor}{\rfloor}
\newcommand{\ins}{\operatorname{Insert}}

\usepackage[
    backend=biber,
    style=alphabetic,
    minalphanames=3,
    maxalphanames=4,
    maxbibnames=99,
    maxcitenames=99,
    doi=true,
    giveninits=true
]{biblatex}
\usepackage{lineno}

\title{Compressing the Lawler-Moore Dynamic Program for Weighted Tardy Jobs}
\author{Klaus Jansen\thanks{Kiel University, kj@informatik.uni-kiel.de, \url{https://orcid.org/0000-0001-8358-6796}. This author was funded by the Deutsche Forschungsgemeinschaft (DFG, German Research Foundation) – Project Number 453769249.}\and Björn Schumacher\thanks{Kiel University, bsch@informatik.uni-kiel.de, \url{https://orcid.org/0009-0002-2448-4794}. This author was funded by the Deutsche Forschungsgemeinschaft (DFG, German Research Foundation) – Project Number 528381760.}\and Corinna Wambsganz\thanks{Kiel University, cwa@informatik.uni-kiel.de, \url{https://orcid.org/0009-0002-4820-9017}. This author was funded by the Deutsche Forschungsgemeinschaft (DFG, German Research Foundation) – Project Number 453769249.}}
\AfterEndEnvironment{abstract}{%
}
\date{}

\begin{document}

\maketitle
\begin{abstract}
  The classic scheduling problem of minimizing the total weight of tardy jobs ($P_m||\sum_jw_jU_j$) asks for a schedule of jobs each with a processing time, weight, and due date on a constant number of parallel identical machines.
  A job is late in a schedule if it completes after its due date, and an optimal schedule is one where the total weight of all late jobs is minimized.
  The classic dynamic programming algorithm of Lawler and Moore [Manag. Sci., 1969] solves this problem in $O(nd_{\max}^m)$ where $n$ is the number of jobs, $d_{\max}$ is the maximum due date of any job, and $m$ is the number of identical machines.
  We show that it is sufficient to only compute part of the dynamic program and obtain the running time $\tilde{O}(n + d_{\max}^{m + 1})$ where $\tilde{O}$ suppresses polylogarithmic factors.
  This affirmatively answers the open question posed by Hermelin, Molter, and Shabtay [INFORMS J. Comput., 2024], whether a running time of $\tilde{O}(n + d_{\max}^2)$ is possible on a single machine.
\end{abstract}
\newpage
\section{Introduction}
Minimizing the total weight of tardy jobs is a classic scheduling problem.
We are given \(n\) jobs with processing times \(p_j\in\bZ_{\geq 1}\), weights \(w_j\in\bZ_{\geq 1}\), and due dates \(d_j\in\bZ_{\geq 1}\) such that \(p_j \leq d_j\) as well as \(m\) parallel identical machines.
A schedule \(\tau\) for this problem is an assignment of jobs to machines together with an order of execution on every machine.
Let \(C_j(\tau)\) be the completion time of job \(j\) in schedule \(\tau\).
A job \(j\) is on time in \(\tau\) if \(C_j(\tau) \leq d_j\) and late otherwise.
We set \(U_j(\tau) = 0\) if \(j\) is on time and \(U_j(\tau) = 1\) if it is late.
The goal is to find a schedule \(\tau\) that minimizes the total weight of late jobs
\[
  \sum_{j=1}^nw_jU_j(\tau).
\]
On a single machine the problem is denoted by \(1||\sum_jw_jU_j\) and on \(m\in O(1)\) identical parallel machines by \(P_m||\sum_jw_jU_j\) using the three-field notation by \textcite{GRAHAM1979287}.

Even on a single machine the problem is NP-hard as it is one of the original 21 problems shown to be NP-hard by \textcite{karp1972}.
For an exact solution it is equivalent to maximize the weight of the on-time jobs instead.
Thus, \(1||\sum_jw_jU_j\) is closely related to the knapsack problem as the problems are the same in the case where all jobs share the same due date.

The classic algorithm for this problem on a single machine is the dynamic programming algorithm due to \textcite{lawlermoore} with running time \(O(n\cdot d_{\max})\) where \(d_{\max}\) is the maximum due date of any job.
It is similar to the algorithm by \textcite{Bellman1957} for knapsack that has a running time of \(O(n\cdot t)\) where \(t\) is the capacity of the knapsack.
Both algorithms follow the same principle: For a job index \(j\in[n]\) and a point in time \(t\in[d_{\max}]\) the dynamic program stores the largest possible weight of on-time jobs with total processing time \(t\) while only using jobs from the set \(\set{1, \ldots, j}\).
In contrast to \citeauthor{Bellman1957}'s algorithm, the algorithm still needs to handle the different due dates.
To cope with these, the jobs are ordered non-decreasingly by due dates before filling the dynamic program table.
This is sufficient since there is always an optimal solution where the on-time jobs are ordered non-decreasingly according to their due dates, also known as the earliest-due-date (EDD) order.
The algorithm is depicted in \cref{alg:lawler-moore}.
It generalizes to \(P_m||\sum_jw_jU_j\) with running time \(O_m(n \cdot d_{\max}^m)\) where \(O_m\) suppresses multiplicative factors depending only on \(m\).

We modify the algorithm by \citeauthor{lawlermoore} by restricting the range of values \(t\) may take in line~\ref{alg:lawler-moore:t-range} of \cref{alg:lawler-moore} and updating the same row in every iteration.
To justify this we give all jobs a priority based on their weight and their index.
We show that there is always a canonical optimal schedule that uses high-priority jobs.
Combining this with the insight that a schedule may contain at most \(\floor{\frac{d_{\max}}{p}}\) on-time jobs with processing time \(p\), we can restrict the necessary updates.
Generalizing this to multiple machines yields the following result.
\begin{restatable}{theorem}{mainthm}\label{thm:main}
  Given an instance of \(P_m||\sum_jw_jU_j\), the optimal value can be calculated in \(\tilde{O}_m(n + d_{\max}^{m+1})\) time and \(O_m(n + d_{\max}^m)\) space.
  Using additional \(\tilde{O}_m(d_{\max}^{m + 1})\) space, we can construct an optimal schedule in the same time.
\end{restatable}
\begin{algorithm}[t]
\caption{Algorithm by Lawler and Moore for \(1||\sum_jw_jU_j\)}
\label{alg:lawler-moore}
\begin{algorithmic}[1]
\Require \(n\) jobs each with $(p_j,d_j,w_j)$
\Ensure Maximum total weight of on-time jobs
\State Initialize $f_j[t] \gets - \infty$ for every \(j\in\set{0, \ldots, n}\) and $t \in\set{0, \ldots, d_{\max}}$ and $f_0[0]\gets0$
\For{$j=1,\ldots,n$ in non-decreasing order by due dates}
  \For{$t=0, 1,\ldots, d_{\max}$} \label{alg:lawler-moore:t-range}
    \If{\(p_j \leq t \leq d_j\)}
      \State \(f_j[t] \gets \max\{f_{j - 1}[t], f_{j - 1}[t - p_j] + w_j\}\)
    \Else
      \State \(f_j[t] \gets f_{j - 1}[t]\)
    \EndIf
  \EndFor
\EndFor
\State \Return $\max_{t \in \set{0, \ldots, d_{\max}}}f_n[t]$
\end{algorithmic}
\end{algorithm}
This presents an improvement over the Lawler-Moore algorithm when \(n\) is larger than \(d_{\max}\) up to logarithmic factors. Together, we obtain a running time of \(\tilde{O}_m(n + d_{\max}^m\min\{d_{\max},n\})\).
We also positively answer the open question posed by \textcite{DBLP:journals/informs/HermelinMS24}, whether there exists an algorithm for \(1||\sum_jw_jU_j\) with running time \(\tilde{O}(n + d_{\max}^2)\).
Moreover, this result is an extension of the folklore reduction for knapsack to reduce the number of items to \(\tilde{O}(t)\) mentioned in Remark 1.5 of \cite{DBLP:conf/icalp/0001RW21}, which also yields an \(\tilde{O}(n + t^2)\) algorithm for knapsack.
The extension is not straightforward since jobs may have distinct due dates, which complicates a simple exchange argument.
We cope with this difficulty by introducing a priority order on the jobs and considering the highest-priority jobs that may be on time for every point in time.

\subsection{Related Work}
The dependency on \(d_{\max}\) in the running time of the Lawler-Moore algorithm for \(1||\sum_jw_jU_j\) cannot be improved without incurring an exponential dependence on \(n\):
Assuming the SETH, for any \(\epsilon > 0\) there is a \(\delta > 0\) such that there is no algorithm with running time \(O(2^{\delta n}(d_{\max} + P)^{1 - \epsilon})\) where \(P\) is the sum of all processing times \cite{DBLP:conf/sofsem/JansenK23}.
Under the \((\min, +)\)-convolution hypothesis there is also no \(\tilde{O}((n + d_{\max})^{2 - \epsilon})\) algorithm \cite{DBLP:journals/talg/CyganMWW19}.
Hence, on a single machine the quadratic dependence on \(d_{\max}\) in \cref{thm:main} is necessary in combination with the near-linear dependence on \(n\). On multiple machines we do not know whether the \(m+1\) exponent is optimal.
The running time of the Lawler-Moore algorithm can also be stated as \(O(n\cdot P^m)\).
When the maximum processing time \(p_{\max}\) is smaller than \(\sqrt{P}\) and \(m \geq 2\), \textcite{DBLP:journals/corr/abs-2604-13642} showed that the running time can be improved to \(O(n\cdot P \cdot p_{\max}^{2m - 2}) \leq O(n^2 \cdot p_{\max}^{2m - 1})\).

Under the \(\forall\exists\)-SETH, for all \(\epsilon > 0\), no algorithm with running time \(\tilde{O}(n + p_{\max}\cdot n^{1 - \epsilon})\) is possible for \(1||\sum_jw_jU_j\) \cite{DBLP:journals/jcss/AbboudBHS22}.
Their reductions also rule out any running time of \(\tilde{O}(n + p_{\max}^c)\) for any constant \(c > 0\) as noted in \cite{DBLP:journals/corr/abs-2604-13642} and thus we cannot hope for a running time of \(\tilde{O}(n + p_{\max}^2)\), which is possible for knapsack \cite{DBLP:conf/stoc/Jin24,DBLP:conf/stoc/Bringmann24}.
Also in the special case of minimizing the tardy processing time on a single machine (\(1||\sum_jp_jU_j\)), where \(w_j = p_j\), a running time of the form \(\tilde{O}(n + p_{\max}^{O(1)})\) is possible.
In particular, \textcite{DBLP:conf/soda/Klein0R23} developed an \(\tilde{O}(n + p_{\max}^3)\) algorithm.
\textcite{DBLP:journals/informs/HermelinMS24} showed that \(1||\sum_jw_jU_j\) can be solved in \(\tilde{O}(n + d_{\#}d_{\max}^2)\) and \(\tilde{O}(d_{\#}n + d_{\#}d_{\max}p_{\max})\) where \(d_{\#}\) is the number of distinct due dates and \(p_{\max}\) is the maximum processing time.
The second summand in both of these running times can be bounded by \(d_{\max}^3\), which leaves the gap that this paper closes.

For the special case \(P_m||\sum_jp_jU_j\) improvements over the Lawler-Moore algorithm are known:  the problem can be solved in \(\tilde{O}(P^m)\) \cite{DBLP:journals/theoretics/FischerW25} and \(O(n\cdot p_{\max}^{O(m)})\) \cite{DBLP:conf/soda/Klein0R23}.
A recent lower bound shows that for any \(\epsilon > 0\) there is no \(2^{o(n)}\cdot P^{m - \epsilon}\) algorithm for \(P_m||\sum_jp_jU_j\) assuming the SETH \cite{DBLP:conf/stoc/BringmannDW26}.

Due to its connection to knapsack, \(1||\sum_jw_jU_j\) is NP-hard even when \(d_{\#} = 1\) where \(d_{\#}\) is the number of different due dates.
Parameterized solely by the number of different processing times \(p_{\#}\) or the number of different weights \(w_{\#}\), the problem is \(W[1]\)-hard \cite{DBLP:journals/jcss/HeegerH26} but in XP \cite{DBLP:journals/anor/HermelinKPS21}.
In contrast, the problem is FPT parameterized by any two of \(\set{d_{\#}, p_{\#}, w_{\#}}\) \cite{DBLP:journals/anor/HermelinKPS21}.

The fastest known exponential algorithms for this are \(O^{*}(2^{\frac{1}{2}n})\) \cite{lente:hal-00944382} and \(O^{*}(2^{n})\) \cite{DBLP:conf/esa/DharDGGW26} for \(1||\sum_jw_jU_j\) and \(P||\sum_jw_jU_j\) (where \(m\) is part of the input), respectively.
\(O^{*}\) suppresses polynomial factors.

\paragraph{Organization of the Paper}
First we establish the existence of an optimal canonical schedule in \cref{sec:schedule-structure}.
Then, in \cref{sec:algorithm} we show how to use the canonical structure algorithmically.
We conclude with an overview of future work in \cref{sec:conclusion}.

\section{Structure of an Optimal Schedule}\label{sec:schedule-structure}
We assume that the jobs are given in EDD order, i.e.,\ \(d_1 \leq \cdots \leq d_n\).
For our application this is not a restriction as such an order can be calculated in \(O(n\log n)\) or \(O(n + d_{\max})\) time using standard sorting algorithms. We define \([k] \coloneqq \set{1, 2, \ldots, k}\) for \(k \in \mathbb{N}\).

Notice that on a single machine at most \(\floor{\frac{d_{\max}}{p}}\) jobs with processing time \(p\) may be on time.
Thus, in every schedule at most \(\kappa_p \coloneqq m\floor{\frac{d_{\max}}{p}}\) jobs with processing time \(p\) are on time.
In contrast to knapsack we cannot simply restrict the set of jobs with processing time \(p\) to the \(\kappa_p\) most profitable ones since not all jobs share the same due date.
To avoid violating the due dates we consider the most profitable jobs for every point in time.
By fixing a global order we can show that every job may only be considered in a certain time interval.
Further, we prove that the total length of these intervals is bounded by \(\tilde{O}_m(d_{\max}^2)\).

We start by defining a total order \(>_w\) for jobs based on the weights of the jobs; ties are broken based on the index.
More precisely, the jobs are ordered non-increasingly and lexicographically by \((w_j, j)\).
Let \(T_{p, t} \coloneqq \set{j\in[n]}[p_j = p, d_j \geq t]\) be the set of jobs with processing time \(p\) that may be on time when finishing at \(t\in\set{p, \ldots, d_{\max}}\).
We restrict the set \(T_{p, t}\) to the \(\min\set{\kappa_p, \abs{T_{p, t}}}\) most profitable jobs according to \(>_w\) and denote this set of jobs by \(\overline{T}_{p, t}\).
For a job \(j\in[n]\), we define \(\overline C_j \coloneqq \set{t\in\set{p_j, \ldots, d_{\max}}}[j\in\overline{T}_{p_j, t}]\) to be the set of canonical completion times of the job.
Clearly, we have
\begin{equation}
  \label{eq:decreasing-Tsets}
  T_{p, p} \supseteq T_{p, p+1} \supseteq \cdots \supseteq T_{p, d_{\max}}
\end{equation}
for every \(p\in[d_{\max}]\).
This directly implies that every \(\overline C_{j}\) is an interval.
The structure of the intervals is depicted in \cref{fig:intervals}.
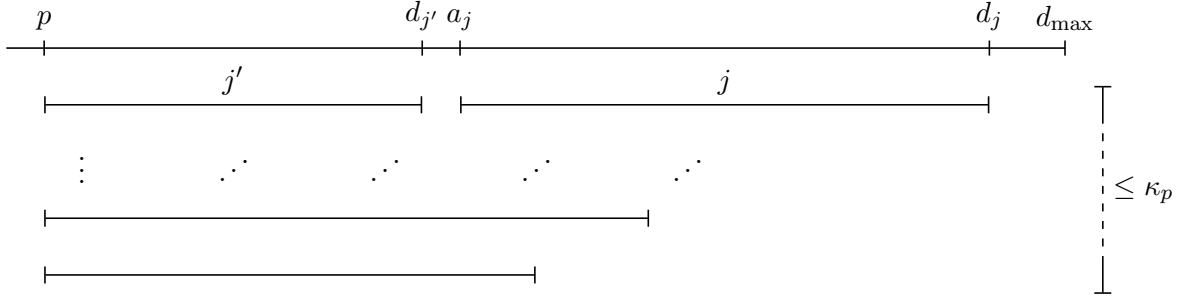
\begin{figure}[t]
  \centering
  \begin{tikzpicture}[semithick]
    \draw (0, 0) -- +(14, 0);
    \draw (.5, -.1) -- +(0, .2) node[above] {\(p\)};
    \draw (5.5, -.1) -- +(0, .2) node[above] {\(d_{j'}\)};
    \draw (6, -.1) -- +(0, .2) node[above] {\(a_j\)};
    \draw (13, -.1) -- +(0, .2) node[above] {\(d_j\)};
    \draw (14, -.1) -- +(0, .2) node[above] {\(d_{\max}\)};
    \draw[|-|] (.5, -.75) -- node[above] {\(j'\)} +(5, 0);
    \draw[|-|] (6, -.75) -- node[above] {\(j\)} +(7, 0);

    \node at (1, -1.5) {\(\vdots\)};
    \node at (3, -1.5) {\(\iddots\)};
    \node at (5, -1.5) {\(\iddots\)};
    \node at (7, -1.5) {\(\iddots\)};
    \node at (9, -1.5) {\(\iddots\)};

    \draw[|-|] (.5, -2.25) -- +(8, 0);
    \draw[|-|] (.5, -3) -- +(6.5, 0);

    \draw[|-] (14.5, -.5) -- +(0, -.5);
    \draw[dashed, dash expand off] (14.5, -1) -- node[right] {\(\leq \kappa_p\)} +(0, -1.75);
    \draw[-|] (14.5, -2.75) -- +(0, -.5);
  \end{tikzpicture}
  \caption{Depiction of the overlapping intervals \(\overline C_j\) for the jobs with processing time \(p\).
    The interval of job \(j\) starts at \(a_j = d_{j'} + 1\) because job \(j'\) is more profitable according to \(>_w\) and ``pushes'' \(j\) out of \(\overline T_{p, d_{j'}}\) since there are already \(\kappa_p\) overlapping intervals.
    The figure also illustrates why the total length of the intervals is bounded by \(d_{\max} \cdot \kappa_p\).}
  \label{fig:intervals}
\end{figure}
\begin{lemma}\label{prop:canonic-is-intervall}
  For every \(j\in[n]\), there is an \(a_j\in\set{p_j, p_j + 1, \ldots, d_j + 1}\) such that \(\overline C_j = \set{a_j, a_j + 1, \ldots, d_j}\).
  If \(\overline{C}_j = \emptyset\), \(a_j = d_j + 1\).
\end{lemma}
\begin{proof}
  Let \(j\in[n]\).
  For every \(t\in\set{p_j,\ldots,d_j - 1}\), we have that \(j\in\overline T_{p_j, t}\) implies \(j\in\overline T_{p_j, t + 1}\) by \cref{eq:decreasing-Tsets}.
  Thus, let \(a_j\) be the smallest point in time such that \(j\in \overline T_{p_j, a_j}\). By the above property we have that \(j\in \overline T_{p_j, t}\) for all \(t \in \set{a_j, a_j + 1, \ldots, d_j}\).
  If no such \(a_j\) exists, i.e., \(\overline C_j = \emptyset\), we set \(a_j = d_j + 1\) as \(\overline{C}_j\) is the empty set.
\end{proof}

\begin{theorem}\label{prop:optimal-structure}
  There exists an optimal schedule \(\tau\) such that every on-time job  completes at a canonical completion time, i.e.,\ \(U_j = 0\) implies that \(C_j(\tau) \in \overline{C}_j\) for every \(j\in[n]\).
  Furthermore, all on-time jobs are ordered in the EDD order given by the indices.
\end{theorem}
\begin{proof}
  Take an optimal schedule that schedules the on-time jobs in index EDD order before it schedules the late jobs on every machine. Choose such an optimal schedule \(\tau\) that also maximizes the sum of on-time indices. Thus, \(\tau\) lexicographically maximizes
  \[
    \sum_{\substack{j\in[n]\\U_j(\tau) = 0}} (w_j, j).
  \]

  Towards a proof by contradiction, suppose that there exists a \(j\in[n]\) such that \(U_j(\tau) = 0\) and \(C_j(\tau) \notin \overline C_j\).
  Then there exists a job \(j'\in\overline T_{p_j, C_j(\tau)}\) with \(U_{j'}(\tau) = 1\) since \(\abs{\overline T_{p_j, C_j(\tau)}} = \kappa_{p_j}\) and at most \(\kappa_{p_j}\) jobs with processing time \(p_j\) are on time in \(\tau\), including \(j\).
  We swap \(j'\) and \(j\) in \(\tau\) and restore the EDD order of the on-time jobs. Restoring the EDD order cannot cause an on-time job to be late as \(p_j = p_{j'}\) and \(d_{j} \leq d_{j'}\). Since \(j' >_w j\), this either increases the weight of the on-time jobs in the schedule or leaves the weight unchanged and increases the sum of indices; both cases yield a contradiction to the choice of \(\tau\).
  Hence, \(C_j(\tau) \in \overline C_j\) for all \(j\in[n]\) with \(U_j(\tau) = 0\).
\end{proof}

Next we bound the total length of the canonical intervals.
\begin{lemma}\label{lem:sumcj}
  It holds that
  \[
    \sum_{j=1}^n \abs{\overline{C}_j} = \sum_{p=1}^{d_{\max}}\sum_{t=p}^{d_{\max}} \abs{\overline T_{p, t}} \leq m d^2_{\max}(1 + \ln d_{\max}) \leq \tilde{O}_m(d_{\max}^2).
  \]
\end{lemma}
\begin{proof}
  Let \(p\in[d_{\max}]\).
  We have
  \[
    \sum_{\substack{j\in[n]\\ p_j = p}} \abs{\overline C_j} = \sum_{t=p}^{d_{\max}}\abs{\overline T_{p, t}}
  \]
  by the definition of the \(\overline C_j\).
  Summing over all processing times shows the first equality.

  Further, we have \(\sum_{t=p}^{d_{\max}}\abs{\overline T_{p, t}} \leq d_{\max}\kappa_p = d_{\max}m\floor{\frac{d_{\max}}{p}} \leq md_{\max}^2 \frac{1}{p}\).
  Thus,
  \[
    \sum_{p=1}^{d_{\max}}\sum_{t=p}^{d_{\max}} \abs{\overline T_{p, t}}
    \leq \sum_{p=1}^{d_{\max}}md_{\max}^2 \frac{1}{p}
    = md_{\max}^2\sum_{p=1}^{d_{\max}} \frac{1}{p}
    \leq md_{\max}^2(1 + \ln d_{\max})
  \]
  where the last inequality holds because the \(k\)th harmonic number is at most \(1 + \ln k\) (cf.\ Appendix A.2 in \cite{DBLP:books/daglib/0023376}).
\end{proof}

\section{Improved Algorithm}\label{sec:algorithm}

In this section, we prove \Cref{thm:main} and present our algorithm, which is an improved version of \Cref{alg:lawler-moore} by Lawler and Moore. The algorithm is a dynamic program that iterates over all jobs and computes the maximum weight of an EDD schedule that uses only the jobs considered up to that point and in which the on-time jobs complete at canonical completion times. The maximum weight of such an EDD schedule for all jobs corresponds to an optimal solution for the \(P_m||\sum_jw_jU_j\) problem. In \Cref{sec:intervals} we show how to obtain the canonical completion times; in \Cref{sec:algo} we present our algorithm and in \Cref{sec:optimal} we prove how to construct an optimal solution to \(P_m||\sum_jw_jU_j\) by using our algorithm.

\subsection{Computation of the Intervals \(\overline C_j\)}\label{sec:intervals}
The values \(a_j,j\in [n],\) defining the intervals \(\overline C_j\) in \Cref{prop:canonic-is-intervall} can be calculated by iterating over the jobs in reverse EDD order and using a heap to keep track of the \(\kappa_p\) most profitable items according to \(>_w\) for every processing time \(p\).
More precisely, we first build lists \(G_p = \set{j\in[n]}[p_j = p]\) of all jobs with the same processing time in increasing EDD index order.
For each of those lists, we do the following two steps: Initialize an empty min-heap where we store jobs with the pair $(w_j,j)$ as their key and sort them according to \(>_w\).
Now we iterate over the jobs in reversed order.
For job \(j\) we insert \((w_j, j)\) into the heap.
If the number of heap entries exceeds \(\kappa_p = m \floor{\frac{d_{\max}}{p}}\), we remove the job with the smallest key.
Let this job be \(h\).
Note that we may have \(h = j\).
We set \(a'_h = d_j + 1\).
After iterating over all jobs, we set \(a'_j = p\) for all jobs still in the heap.
\begin{lemma}\label{lem:correctai}
    The values \(a_j'\) computed by the procedure above are the correct values \(a_j\) for all \(j \in [n]\).
\end{lemma}
\begin{proof}
  Let \(h\in[n]\).  We show \(a_h' = a_h\) by showing the inequality in both directions.
  \begin{description}
  \item[\(a'_h \leq a_h\)] If \(h\) is never removed from the heap, we have \(a'_h = p_h \leq a_h\) and we are done.
    Now assume that \(h\) was removed.
    Then there exists a job \(j\) whose addition to the heap caused the removal of \(h\).
    When \(j\) is added, all jobs currently in the heap have a due date of at least \(d_j\) and a higher priority than \(h\).
    Therefore \(h\notin \overline{T}_{p_h, d_j}\) and \(a'_h = d_j + 1 \leq a_h\).
  \item[\(a_h \leq a'_h\)] If \(a_h = p_h\), the job \(h\) is one of the \(\kappa_{p_h}\) most profitable jobs according to \(>_w\) and thus still in the heap after all jobs are processed, which implies \(a'_h = p_h\).
    Now assume that \(p_h < a_h\).
    After all jobs in \(T_{p_h, a_h - 1}\) are processed, \(h\) cannot be in the min-heap anymore since \(\overline{T}_{p_h, a_h - 1}\) has \(\kappa_{p_h}\) entries and does not contain \(h\), which implies that \(h\) was removed before jobs with due dates smaller than \(a_h - 1\) are considered.
    Therefore, \(a'_h > a_h - 1\), which means \(a'_h \geq a_h\). \qedhere
  \end{description}
\end{proof}

\begin{lemma}\label{lem:computeIntervals}
    The values \(a_1, \ldots, a_n\) for the intervals \(\overline C_1, \ldots, \overline C_n\) can be computed in \(O(n\log n)\) time using \(O(n)\) space.
\end{lemma}
\begin{proof}
  By \Cref{lem:correctai} the procedure above computes the values \(a_1, \ldots, a_n\) for the intervals \(\overline C_1, \ldots, \overline C_n\).
  We analyze its running time and space consumption.
  We can build the lists \(G_p\) in \(O(n)\) time using \(O(n)\) space.
  For each job, we have at most one insertion and one deletion operation on the min-heap.
  The min-heap has size at most \(n\).
  Therefore, these operations can be done in \(O(n \log n)\) time and \(O(n)\) space (cf.\ \cite{DBLP:books/daglib/0023376}).
  All update operations for the values \(a_j, j \in [n],\) can also be done in \(O(n)\) time using \(O(n)\) space.
  Together, this yields a running time of \(O(n\log n)\) and a space consumption of \(O(n)\).
\end{proof}

\subsection{The Algorithm}\label{sec:algo}
In the following, we write \(\vect e_i\) for the \(i\)th unit vector.
First we give the formal recurrence for the classic Lawler-Moore algorithm and afterwards the changes we make for our algorithm.

We define values \(F_j(\vect t)\) as the maximum weight of an EDD schedule using a subset of the jobs \(\set{1, \ldots, j}\) with machine-load vector \(\vect t\).
If no such schedule exists, its value is $-\infty$. For \(j=0\) and \(\vect t \in \set{0,1, \ldots, d_{\max}}^m \setminus \set{\mathbf{0}}\), we initialize
\[F_0(\mathbf{0})=0, \qquad F_0(\vect t) = -\infty.\]
For \(j \in \set{1, \ldots ,n}\) and \(\vect t \in \set{0, \ldots, d_{\max}}^m\), we set
\begin{equation}\label{eq:DP}
 F_j(\vect t)=
 \max\{F_{j-1}(\vect t),\ \max_{\substack{r \in [m] \\ p_j \leq t_r \leq d_j}}\{F_{j-1}(\vect t-p_j\vect e_r)+w_j\}\}.
\end{equation}
We use \(-\infty + w_j = -\infty\) and \(F_j(\vect t) = -\infty\) for every \(t \notin\set{0, \ldots, d_{\max}}^m\) to simplify the notation.

Compared to the classic algorithm, our algorithm stores only a single row of the dynamic program and not a row per job.
For a job, we first collect all updates and then apply them afterwards to avoid using any job more than once.
Further, we restrict the value of \(t_r\) in \cref{eq:DP} to values from \(\overline C_j\).
\cref{prop:optimal-structure} guarantees that we still find the optimal total weight of on-time jobs.

For the algorithm, we define a function
\[
  \ins_r\colon
  \begin{cases*}
    \set{0,\ldots,d_{\max}}^{m-1} \times \{0,\ldots,d_{\max}\} \rightarrow \set{0,\ldots,d_{\max}}^m, \\
    (\vect u, t) \mapsto (\vect u_1, \ldots, \vect u_{r - 1}, t, \vect u_r, \ldots ,\vect u_{m - 1})
  \end{cases*}
\]
for each \(r \in [m]\), which constructs from a vector \(\vect u \in \set{0,\ldots,d_{\max}}^{m-1}\) a new vector \(\vect x \in \set{0,\ldots,d_{\max}}^m\) by placing the value \(t\) as its \(r\)th coordinate and copying the old values from \(\vect u\) as its other coordinates.
The complete algorithm is given as \cref{alg:main}.
The array \(f\) in the algorithm represents the dynamic computation.
It is only modified at the end of every iteration of the outermost loop by applying the updates from the corresponding list.

\begin{algorithm}[htbp]
\caption{WeightedTardyJobs}
\label{alg:main}
\begin{algorithmic}[1]
\Require All original jobs $(p_j,d_j,w_j)$ in EDD order, $m$ identical parallel machines
\Ensure Maximum total weight of on-time jobs
\State Compute the values $(a_1,\ldots,a_n)$ as described in \Cref{sec:intervals}
\State Initialize $f[\vect x] \gets - \infty$ for every $\vect x \in L \coloneq \set{0, \ldots, d_{\max}}^m$
\State $f[\mathbf{0}]\gets0$
\For{$j=1,\ldots,n$ in increasing order}
  \State Updates $\gets []$
  \For{$t=d_j,d_j-1,\ldots,a_j$ in decreasing order}
    \For{$\vect u \in \set{0, \ldots, d_{\max}}^{m-1}$}
      \For{$r=1, \ldots, m$}
        \State $\vect x \gets \ins_r(\vect u,t)$
        \State $\vect y \gets \vect x - p_j \vect e_r$
        \If{$f[\vect y]\ne-\infty$}
          \State Updates.append(($\vect x$,$f[\vect y] + w_j$)) \label{alg:main:add-update}
        \EndIf
      \EndFor
    \EndFor
  \EndFor
  \For{($\vect x$, best) in Updates}
    \State $f[\vect x] \gets \max\set{f[\vect x], \text{best}}$
  \EndFor
\EndFor
\State \Return $\max_{\vect x \in L}f[\vect x]$
\end{algorithmic}
\end{algorithm}

\begin{lemma}\label{lem:algcor}
    Given an instance \(I\) of \(P||\sum_jw_jU_j\), \Cref{alg:main} computes a value \(w\) such that \(\sum_{j \in [n]}w_j - w \leq \text{OPT}(I)\) where \(\text{OPT}(I)\) denotes the optimal value of the instance \(I\).
\end{lemma}
\begin{proof}
By \Cref{prop:optimal-structure} we know that there exists an optimal schedule \(\tau\) such that every on-time job completes at a canonical completion time and such that all on-time jobs are ordered in EDD order given by the indices. \Cref{lem:correctai} guarantees that the values \(a_1,\ldots,a_n\) in \Cref{alg:main} are correct. Since \Cref{alg:main} iterates for every job \(j\) through all canonical completion times \(\overline{C}_j\) and determines for every possible machine-load vector the maximum weight, \Cref{alg:main} also considers the same placement of \(j\) as in the optimal schedule \(\tau\). With this, we obtain \(\sum_{j \in [n]}w_j - w \leq \text{OPT}(I)\).
\end{proof}

\begin{lemma}\label{lem:algrt}
    \Cref{alg:main} has a running time of \(\tilde{O}_m(n+d_{\max}^{m+1})\) and uses \(O_m (n + d_{\max}^m)\) space.
\end{lemma}
\begin{proof}
    The values \(a_1, \ldots, a_n\) can be computed in \(O(n\log n)\) time using \(O(n)\) space by \Cref{lem:computeIntervals}. The initialization of \(f\) runs in \(O((d_{\max}+1)^m)\) time and space. Afterwards, \Cref{alg:main} considers all possible canonical completion times and all vectors \(\vect u \in \set{0,\ldots,d_{\max}}^{m-1}\) for every machine. The function \(\ins_r(\vect u,t)\) takes \(O(m)\) time and space. All other operations done by \Cref{alg:main} are simple comparisons, computations, array accesses and modifications that run in \(O(1)\) time. In total, we obtain a running time of
    \[
    O(n\log n + m \cdot \sum_{j=1}^n{\abs{\overline{C}_j}} \cdot (d_{\max}+1)^{m-1} \cdot m) \leq \tilde{O}_m(n + d_{\max}^{m+1})
    \]
    where the inequality holds by \Cref{lem:sumcj}.
    The list Updates takes at most \(O_m(d_{\max}^m)\) space while all updates performed in \(f\) are in-place. Altogether, \Cref{alg:main} uses \(O_m(n+d_{\max}^m)\) space.
\end{proof}

\subsection{Finding an Optimal Schedule}\label{sec:optimal}
Constructing an optimal schedule from the dynamic program cannot be done by simple backtracking due to the restriction to a single row.
To be able to construct a schedule without constructing the entire \(n\cdot (d_{\max}+1)^m\) table, we store a sparse pointer data structure.
The number of necessary back links can be bounded in the same way as the number of updates we perform to calculate the table.
\begin{lemma}\label{prop:schedule-construction}
  There exists a schedule where the weighted number of on-time jobs is at least the output of \cref{alg:main}.
  Using additional \(\tilde{O}_m(d_{\max}^{m + 1})\) space, we can construct a schedule in time \(O_m(n + d_{\max}^m)\) after running the dynamic program.
\end{lemma}
\begin{proof}
  Let \(f_0\) be the initial state of \(f\) in \cref{alg:main} and \(f_i\) the state after the \(i\)th job was processed.
  Let \(\vect x \in L\) be such that \(f_n(\vect x) = \max_{\vect y\in L}\set{f_n(\vect y)}\).
  By the recursive formula of the dynamic program every entry in \(f_i\) has a predecessor in \(f_{i - 1}\).
  We denote the predecessors of \(\vect x\) by
  \[
    \vect 0 = \vect x^{(0)}, \vect x^{(1)}, \ldots, \vect x^{(n)} = \vect x.
  \]
  When \(\vect x^{(i - 1)} \ne \vect x^{(i)}\), then there exists an \(r_i\in[m]\) such
  that \(\vect x^{(i - 1)} + p_i\vect e_{r_i} = \vect x^{(i)}\) and \(f_{i - 1}(\vect x^{(i - 1)}) + w_i = f_{i}(\vect x^{(i)})\).
  We construct a schedule as follows: for all indices \(i\in[n]\) with a defined \(r_i\) we schedule job \(i\) on machine \(r_i\).
  Since \Cref{alg:main} only considers canonical completion times for each job \(j\), it holds that \(C_j(\tau') \leq d_j\) for our constructed schedule \(\tau'\).
  We arbitrarily schedule the remaining jobs after the on-time jobs.
  The weight of the on-time jobs is at least \(f_n(\vect x)\), which shows the first part of the claim.

  Next, we discuss how to obtain such a schedule constructively by introducing back links.
  Every entry is associated with a reference to a back link that describes how this value was achieved.
  More precisely, every time we find an improvement over a value in \(f\) in line~\ref{alg:main:add-update} we record the current job index, the machine placement, and a reference to the back link yielding the predecessor.
  During the update process we also update the associated back link for every entry by keeping a reference to the back link from the most profitable update.
  After the dynamic program has finished, the back links can be used to construct a schedule as described above.
  This takes \(O(md_{\max} + n)\) time since at most \(md_{\max}\) back links need to be considered and the remaining jobs need to be scheduled somewhere.

  The number of back links created during the algorithm is bounded by the number of candidates considered in the dynamic program.
  The proof of \cref{lem:algrt} shows that the number of candidates considered is bounded by \(\tilde{O}_m(d_{\max}^{m + 1})\), which directly gives the necessary space bound.
\end{proof}
Now we have all the ingredients needed to prove our main result.
\mainthm*
\begin{proof}
  Combining \cref{lem:algcor,prop:schedule-construction} shows that \cref{alg:main} calculates the optimal value.
  \cref{prop:schedule-construction} shows how to construct the schedule.
  \cref{lem:algrt,prop:schedule-construction} together show the running time and space bounds.
\end{proof}

\section{Conclusion}\label{sec:conclusion}
We showed that \(P_m||\sum_jw_jU_j\) can be solved in time \(\tilde{O}_m(n + d_{\max}^{m + 1})\).
For \(1||\sum_jw_jU_j\) this result implies an algorithm with running time \(\tilde{O}(n + d_{\max}^2)\), which answers the open question by \textcite{DBLP:journals/informs/HermelinMS24}.
Assuming the \((\min, +)\)-convolution hypothesis there is no algorithm with running time \(\tilde{O}((n + d_{\max})^{2 - \epsilon})\) for any \(\epsilon > 0\) due to the connection to the knapsack problem.
In contrast to knapsack it is not known whether a truly subquadratic \((\min, +)\)-convolution algorithm yields a subquadratic algorithm for \(1||\sum_jw_jU_j\).
An interesting open question is whether these problems are equivalent under subquadratic reductions.

For the single-machine problem, it is known that a running time \(\tilde{O}(n + p_{\max}^c)\) for any constant \(c > 0\) is impossible under the \(\forall\exists\)-SETH \cite{DBLP:journals/jcss/AbboudBHS22,DBLP:journals/corr/abs-2604-13642}.
In particular, a running time of \(\tilde{O}(n + p_{\max}^2)\) is ruled out under this assumption, but is a hybrid running time of \(\tilde{O}(n + d_{\max}p_{\max})\) possible?
This would also be a direct improvement over the \(\tilde{O}(d_{\#}n + d_{\#}d_{\max}p_{\max})\) algorithm by \textcite{DBLP:journals/informs/HermelinMS24}, removing both \(d_{\#}\) factors.

\section*{Declaration of Generative AI Use}
The idea for the algorithm was generated by GPT-6 Astra.
The authors independently verified the idea and worked out the details of the algorithm and the proofs.
Afterwards, GPT-6 Astra was used to polish the paper.
The authors take full responsibility for the entire content, technical claims, and integrity of the paper.

\printbibliography

\end{document}